\documentclass[10pt,journal,compsoc]{IEEEtran}

\ifCLASSOPTIONcompsoc
\else
\fi
\ifCLASSINFOpdf
\else
\fi
\usepackage[cmex10]{amsmath}
\usepackage{amsmath,amssymb,amsthm} 
\usepackage{graphics, graphicx}               
\usepackage{times}
\usepackage{subfigure}

\newtheorem{thm}{Theorem}
\newtheorem{exm}{Example}

\newtheorem{prp}{Proposition}

\theoremstyle{remark}

\newtheorem{res}{A}

\theoremstyle{definition}
\newtheorem{defi}{Definition}

\usepackage{url}

\begin{document}
%
\title{Curve Reconstruction in Riemannian Manifolds:\\
Ordering Motion Frames}
%
%
%
%

\author{Pratik Shah and Samaresh Chatterji \\
\small{Dhirubhai Ambani Institute of Information and Communication Technology\\
Gandhinagar, India\\}
E-mail: \small{\{pratik\_shah, samaresh\_chatterji\}@daiict.ac.in}
}

%
%

\markboth{Curve Reconstruction in Riemannian Manifolds}{Pratik Shah \makelowercase{et al.}}
%



\IEEEcompsoctitleabstractindextext{%
\begin{abstract}
In this article we extend the computational geometric 
curve reconstruction approach to curves in Riemannian manifolds.  
We prove that the minimal spanning tree, given a sufficiently
dense sample set of a curve, correctly reconstructs the
smooth arcs and further closed and simple curves in
Riemannian manifolds.  
The proof is based on the behaviour of the curve
segment inside the tubular neighbourhood of the curve.
To take care of the local topological changes of the
manifold, the tubular neighbourhood is constructed in
consideration with the injectivity radius of the 
underlying Riemannian manifold.
We also present examples of successfully reconstructed curves
and show applications of curve reconstruction to 
ordering motion frames.
\end{abstract}


}

\maketitle

\section{Introduction}
\IEEEPARstart{T}{he} curve reconstruction problem can be thought of
as \textit{connect the dots}.  
The idea is quite similar to Nyquist's sampling theorem 
for band limited signals in signal processing.
The only difference is in terms of ordering the sample.  
Unlike in the latter case, the ordering is lost when we have
a sample of data points on a curve.
The problem of reconstruction thus demands first to establish
a proper sampling criterion for the curve, next a provable
ordering algorithm based on the sampling criterion to
give a polygonal approximation of the original curve, and finaly 
an interpolating scheme to smooth out the corners.
The nature of problem to be dealt with in this article corresponds to
the first. 
Suppose an object in $\mathbb{R}^3$ 
is in motion and we have captured some frames of this motion. But
these frames are jumbled up, i.e. the ordering is lost. Reconstruct
the original motion given that the frames captured form a 
dense sample set.  

We extend the ideas of
curve reconstruction in Euclidean space - $\mathbb{R}^n$
to \textit{Riemannian manifolds}.  
In this article we make an attempt to extend the
computational geometric curve reconstruction approach to  
curved spaces.

It turns out that the riemannian manifold we are interested in, $SE(3)$, is
endowed with an additional structure
of the group, which makes Riemannian manifold into a Lie Group.  It is a 
well studied object in
physics and mathematics.  Although no bi-invariant metric exist on $SE(3)$, 
together with the Riemannian metric defined on 
it the \textit{exponential map} and further a left invariant distance metric on $SE(3)$ is 
expressed in a closed form.  We give examples of successfully reconstructed 
curves in $SE(2)$
and $SE(3)$.  We also show an application of curve reconstruction in
$SE(2)$ for video frame ordering.
We show that for a densely sampled curves,
\textit{minimal spanning tree}(MST) gives a correct polygonal reconstruction of
curves in Riemannian manifolds.  We also interpolate the ordered point set 
by a partial geodesic interpolation.  Further we propose to interpolate the ordered
point set with de Casteljau algorithm assuming the boundary conditions are
known.

\section{Background}
We begin with a quick review of curve
reconstruction in the plane, keeping the notations and definitions
as general as possible.  
A curve, for our purpose, is the set of image
points of a function $\gamma:[0,1]\to \mathcal{M}$.
More specifically looking at the application at hand, we 
restrict $\mathcal{M}$ to be a differentiable manifold.  We denote 
a curve by a symbol $\mathcal{C}$.			
Since it is an image of a compact interval, $\mathcal{C}$ is a one
dimensional compact manifold.  It is also differentiable if $\gamma$ 
is differentiable.  $\mathcal{C}$ is smooth if $\gamma$ is smooth.
A subclass of curves those are smooth and simple is of  
vital importance in pattern recognition, graphics, image processing and computer vision.  

If $\mathcal{M}$ is $\mathbb{R}^2$ then the problem is of
reconstructing curves in a plane and $\mathcal{C}$ is
planar.  $\mathbb{R}^2$ along with the
standard Euclidean distance metric becomes a metric space.  Naturally
the question arise, is it always possible to have a finite sample 
set $\mathcal{S}\subset \mathcal{C}$ which captures everything
about $\mathcal{C}$?  The answer lies in the fact that 
$\mathcal{C}$ is compact and $\gamma$ is smooth.  To appreciate
it more clearly let us look at a definition of an $\varepsilon$-net.
If $\varepsilon >0$ is given, a subset $\mathcal{S}$ of $\mathcal{C}$
is called an $\varepsilon$-net if $\mathcal{S}$ is finite and its
$\mathcal{C} \subset \cup_{s\in \mathcal{S}} B_{\varepsilon}(s)$, where
$B_{\varepsilon}(s)$ is an open ball in $\mathcal{M}$ with radius
$\varepsilon$.  In other words if $\mathcal{S}$ is finite and its points
are scattered through $\mathcal{C}$ in such a way that each point
of $\mathcal{C}$ is distant by less than $\varepsilon$ from at least
one point of $\mathcal{S}$.  Since $\mathcal{C}$ is compact every
cover will have a finite sub-cover, which shows a possibility of
a finite representative sample set of $\mathcal{C}$.
The concept of $\varepsilon$-net captures the idea of sampling 
criterion very well.

In \cite{gomes94}, based on the uniform sampling criterion an Euclidean
MST is suggested for the reconstruction.
In the initial phase of the development, uniform sampling criterion was
the bottleneck.  The first breakthrough came with the non-uniform sampling
criterion suggested based on the local feature size by \cite{amenta99}.  Unlike the
uniform sampling it samples the curve more where the details are more.
Non-uniform sampling is based on the medial axis of the curve.
The \textit{medial axis} of a curve $\mathcal{C}$ is closure of the set of points
in $\mathcal{M}$ which have two or more closest points in $\mathcal{C}$.  
A simple closed curve in a plane divides the plane into two disjoint regions.
Medial axis can be thought of as the union of disjoint skeletons of the regions 
divided by the curve.  The \textit{Local feature size}, $f(p)$, of a point 
$p\in \mathcal{C}$ is defined as the Euclidean distance from 
$p$ to the closest point $m$ on the medial axis.

$\mathcal{C}$ is \textit{$\varepsilon$-sampled} by a set of sample points $\mathcal{S}$ 
if every $p\in \mathcal{C}$ is within distance $\varepsilon \cdot f(p)$ of a sample $s\in \mathcal{S}$.  
The algorithm suggested in \cite{amenta99}
works based on voronoi and its dual delaunay triangulation.  All delaunay based approaches
can be put under a single formalism, the restricted delaunay complex, 
as shown in \cite{edelsbrunner98}.  Every approach is similar in construction 
and differs only at how it restricts the delaunay complex.  
The crust \cite{amenta99} and further improvements NN-crust \cite{dey2007}
can handle smooth curves.  In some cases \cite{dey2007} it is possible to 
tackle the curves with boundaries and also curves in $\mathbb{R}^d$, $d$-dimensional
euclidean space.  The CRUST and NN-CRUST assume that the sample $\mathcal{S}$ is derived from
a smooth curve $\mathcal{C}$.  The question of reconstructing non smooth cuves have also been studied.
An extensive experimentation with various curve reconstruction algorithms is
carried out in \cite{althaus}.  In \cite{dey99} an extension of NN-CRUST to $\mathbb{R}^d$
is presented.  Which opens up possibilities of extending the existing 
delaunay based reconstruction algorithms to higher dimensional euclidean spaces.
We show an example of a curve in $SE(2)$ reconstructed by NN-CRUST.

\subsection{Organization of the article}
In this paper, we pose the problem of curve 
reconstruction in higher dimensional
curved spaces.  To author's knowledge there are
no reports of efforts made in this direction.
We pose the problem as follows.
Let $\mathcal{M}$ be a riemannian manifold.  
$\mathcal{C}:[0,1]\to\mathcal{M}$ is a smooth, closed 
and simple curve.  Given a finite sample 
$\mathcal{S}\subset\mathcal{C}$ reconstruct the 
$\mathcal{C}$.  Problem involves defining the
appropriate $\mathcal{S}$, suggesting a provable
algorithm for geodesic polygonal approximation and an
interpolation scheme.  

To deal with such objects, 
we first examine the notion of distance on 
surfaces and then move on to more general manifolds
in section \ref{riemetric}.  We
make the Riemannian manifold into a metric space with the help of
the Riemennian inner product.
Next we examine the metric structure of 
$SE(2)$ and $SE(3)$.
With an example in section \ref{medial} we show that 
the medial axis based sampling criterion becomes meaningless in
curved spaces.  We define the dense sample set of a curve on 
Riemannian manifolds and show that in 
section \ref{order} that it is possible to reorder the
dense sample set.  We present successfully reconstructed
curves in $SE(2)$ and $SE(3)$ in section \ref{sim}.
And finally we conclude giving due acknowledgements.
	
\section{Metric Structure on Riemannian Manifolds}
\label{riemetric}
\subsection{$\mathbb{R}^n$ and a surface in $\mathbb{R}^3$}
A curve\footnote{We restrict out attention to smooth curves, i.e. curves which 
are $C^{\infty}$.} 
in space and a curve on the surface are two different entities.
$\mathbb{R}^n$ can be thought of as a
Riemannian manifold with the usual \textit{vector inner product} as the Riemannian
metric.  The tangent space at a point of $\mathbb{R}^n$ is also
an $n$-dimensional vector space.  With the help of the vector inner product
the length of the curve $x:[0,1]\to \mathbb{R}^n$ is defined as:
$L(x) = \int_0^1 \sqrt{\langle x'(t),x'(t)\rangle} dt $.
It turns out that the minimum length curve between two points in the 
Euclidean space is a straight line segment connecting them.  So the 
distance between two points in $\mathbb{R}^n$ is given by 
$d(x,y) = \sqrt{\sum_{i=1}^n (x_i- y_i)^2}$.
With this as a metric $(\mathbb{R}^n,d)$ is a metric space.

Now let us look at a two dimensional surface $\mathcal{M}$
embedded in $\mathbb{R}^3$.  \textit{Two dimensions} here indicate that each 
point $p\in \mathcal{M}$ has a neighbourhood \textit{homeomorphic}\footnote{homeomorphic
here can be replaced by \textit{diffeomorphic} for a differentiable manifold.} 
to a subset of $\mathbb{R}^2$.  In other words, if we associate with each point 
$p\in \mathcal{M}$ a tangent space $T_p \mathcal{M}$
 then the dimension of
$T_p \mathcal{M}$ is two\footnote{For notations and
definitions of basic differential geometric terms, we have referred to 
\cite{gray06}.}
, i.e. two linearly independent vectors are
required to span $T_p \mathcal{M}$.  It is now this tangent space
and the basis vectors of this space which decide the Riemannian metric
for a given surface.  Let us consider a surface patch 
$x(u,v)\subset\mathbb{R}^3$ parametrized by 
$(u,v)\in\mathcal{U}\subset\mathbb{R}^2$.  In this case $x$ is our
manifold $\mathcal{M}$.  Riemannian  metric 
is defined as:
\begin{equation}
	g_{ij} = \left[\begin{array}{cc}
				E=\langle x_u,x_u \rangle & F=\langle x_u,x_v \rangle\\
				F=\langle x_v,x_u \rangle & G=\langle x_v,x_v \rangle 
			\end{array}\right]
\end{equation}
where $x_u$ and $x_v$ are partial derivatives of $x(u,v)$ w.r.t. $u$ and $v$
respectively.  Any vector in $T_p\mathcal{M}$ can be expressed in terms of
these basis vectors $x_u$ and $x_v$.  The inner product for vectors $x_1,x_2 \in T_p\mathcal{M}$
is given by $\langle x_1,x_2 \rangle_g = x_1^T g_{ij} x_2$, where $x_1$ and $x_2$ are
column vectors.
Given a curve $\gamma(t)\in\mathcal{M}$, the length of the curve is
defined as :
\begin{equation}
L(\gamma) = \int_0^1 \sqrt{ \langle \gamma'(t),\gamma'(t)\rangle_g} dt 
\end{equation}

Given $p,q\in\mathcal{M} $, let $\gamma$ be a curve lying in $\mathcal{M}$
with $p, q$ as end points.  Then
\begin{equation}
	d(p,q) = \inf L(\gamma)
\end{equation}
is a valid metric on $\mathcal{M}$.  A $\gamma^*$ for which the
distance between two points is minimized is called a \textit{geodesic} curve
on the manifold $\mathcal{M}$.  As we will see in the following example,
even for a simple looking parametrized surface finding a closed form
expression for the geodesic curve is difficult.  In practice,
$\gamma^*$ is obtained by numerical approximations\cite{kimmel98}.
\begin{exm}
Let $x(u,v)=(u,v,u\cdot v)$ which leads to $x_u=(1,0,v)$, $x_v=(0,1,u)$ and
$E = 1 + v^2 $, $F = u \cdot v$ and $G = 1+u^2$.  A curve in $x(u,v)$ is,
$\gamma(t)=x(u(t),v(t))=(u(t),v(t),u(t)\cdot v(t))$.  The length of the
curve $\gamma(t)$, $t\in[t_0,t_1]$ is $\int_{t_0}^{t_1} \sqrt{E u'^2 + 2\cdot F\cdot u'\cdot v'+G v'^2} dt$,
where $u'$ and $v'$ are $du/dt$ and $dv/dt$ respectively.

If we try to minimize the length function by Euler-Lagrange minimization
we get for each of the co-ordinates a second order ordinary non-linear differential equation 
to solve.  In this example these equations are:
\begin{eqnarray}
\frac{d^2u}{dt^2}+2\frac{v}{1+u^2+v^2}\frac{du}{dt}
            \frac{dv}{dt}=0 \\
\frac{d^2v}{dt^2}+2\frac{u}{1+u^2+v^2}\frac{du}{dt}
            \frac{dv}{dt}=0
\end{eqnarray}
Let the boundary points, the points between which we are trying to find the
geodesic distance, be $(1,1,1)$ and $(-1,-1,1)$.  We solve the BVP for the above
system of equations with matlab boundary value solver.
The resultant geodesic and the initial guess are shown in the Figure.\ref{bilinear}.
\begin{figure}[h]
    \begin{center}
        \includegraphics[scale=0.45]{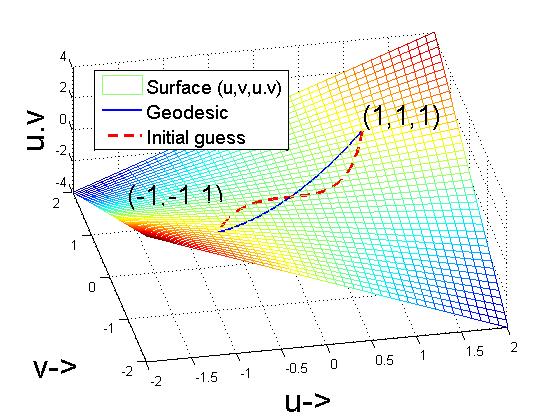}
    \end{center}
    \label{bilinear}
    \caption{A bilinear surface and a geodesic}
\end{figure}
\end{exm}

\subsection{Euclidean Motion Groups}
	\label{metric_emg}
Consider an object in plane undergoing a rigid body euclidean motion.
This motion can be decomposed into a rotation with respect to the
center of mass of the object and a translation of the center of mass
of the object.  All possible configurations of an object in plane
can be represented by $(\theta,u,v)$(i.e. orientation of the 
principle axis and the co-ordinates of the center of mass of the object), where
$0\leq\theta\leq 2\pi$ and $(u,v) \in \mathbb{R}^2$.  Let all such 
configurations form a set $\mathbb{S}$.  It is rather intuitive to define a
metric on $\mathbb{S}$ so as to compare two configurations of an
object.  If $A_1=(\theta_1,u_1,v_1)$ and $A_2=(\theta_2,u_2,v_2)$ be two 
configurations in $\mathbb{S}$ then it is easy to verify that
\begin{equation}
d(A_1,A_2):= \sqrt{a(\theta_1-\theta_2)^2+b(u_1-u_2)^2+b(v_1-v_2)^2}
\label{se2_metric}
\end{equation}
is a valid metric on $\mathbb{S}$ corresponding to 
the riemannian inner product $\langle A_1,A_1\rangle_R = A_1^T R A_1$, and 
$R = \left[\begin{array}{cc} a & 0 \\ 0 & bI_2 \end{array}\right]$ a 
positive definite matrix.  
Moreover for given $A_1,A_2$, left
composition with $A\in\mathbb{S}$, i.e $A(A_1)=(\theta+\theta_1,u+u_1,v+v_1)$,
the above defined metric leads to 
$d(A_1,A_2) = d(A(A_1),A(A_2))$.  Hence we have a left invariant metric
defined on $\mathbb{S}$.  Physical interpretation of the left invariance
is the freedom in choice of the inertial reference frame.  
The matrix representation of $\mathbb{S}$,
the euclidean motion group, is denoted by $SE(2)$.  And a typical
element of $SE(2)$ is made up of a rotation matrix and a translational 
vector.  Correspondence between $\mathbb{S}$ and $SE(2)$ is given by
\[
(\theta,u,v)\Leftrightarrow 
\left[\begin{array}{ccc}
\cos{\theta} & \sin{\theta} & u\\
-\sin{\theta} & \cos{\theta} & v\\
0 & 0 & 1 \end{array}\right]
\]
A typical curve between two configurations in $SE(2)$ 
and the geodesic segment from $A_1$ to $A_2$ are show 
in Fig.\ref{se2_curve_1}.

\begin{figure}[h]
	\begin{center}
		\includegraphics[scale=.5]{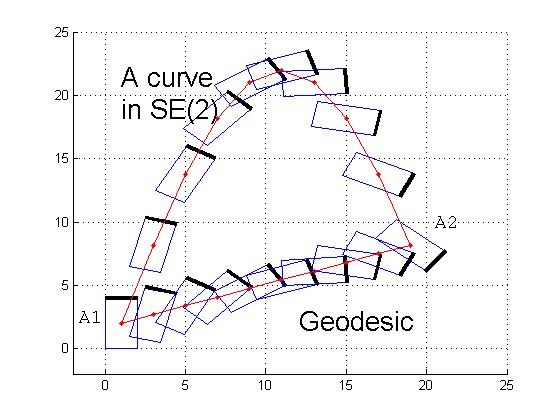}
	\end{center}
	\caption{Comparision of a curve and a geodesic in $SE(2)$ between two configurations $A_1$ and $A_2$.}
	\label{se2_curve_1}
\end{figure}
$SE(2)$ is explored in the domain of image processing for segmentation
in object tracking where
one is interested in constrained evolution of the curve under the action
of $SE(2)$, a lie group. 

In general the group of rigid body motions in 
$\mathbb{R}^n$ is the semi-direct product \cite{selig05} of the special
 orthogonal group with
$\mathbb{R}^n$ itself.
\[SE(n)= SO(n)\ltimes \mathbb{R}^n\]
Unlike $\mathbb{R}^2$ the rotations in $\mathbb{R}^3$ are not commutative.
And that reflects in the group composition of $SO(3)$, $R_1 R_2\neq R_2 R_1$,
$R_1,R_2 \in SO(3)$.
The product of two rigid body motions $(R_1,d_1),(R_2,d_2) \in SE(3)$ is given by
$(R_2,d_2)(R_1,d_1) = (R_2R_1,R_2d_1+d_2)$.
The matrix representation of elements of $SE(3)$
\begin{equation}
SE(3) =  \lbrace A\vert A = \left[\begin{array}{cc} R & d\\0 & 1\end{array}\right],
R\in SO(3), d \in \mathbb{R}^3 \rbrace
\end{equation}

The tangent space at the group identity in $SO(3)$ and $SE(3)$ are
the lie algebras $so(3)$ and $se(3)$ respectively.
\begin{eqnarray}
so(3) = \lbrace [\omega]|[\omega] \in \mathbb{R}^{3\times 3}, [\omega]^T = -[\omega] \rbrace,\\
se(3)=\lbrace S = \left[\begin{array}{cc} [\omega] & v\\ 0 & 0 \end{array}\right], 
[\omega]\in so(3), v \in \mathbb{R}^3  \rbrace
\end{eqnarray}
Where $[\omega]$ is a skew symmetric matrix \cite{gray06} corresponding to the vector 
$\omega=[\omega_x, \omega_y, \omega_z]\in\mathbb{R}^3$.  The $||\omega||$ gives
the amount of rotation with respect to the unit vector along $\omega$.
The exponential map is a diffeomorphism \cite{zefran95} connecting the lie algebra to the 
lie group.  The $\exp : se(3)\to SE(3)$ is given by the usual matrix
exponential as $\exp (S) = \sum_{n=0}^\infty \frac{S^n}{n!}$.

Consider a rigid body moving in free space.  We fix any inertial reference 
frame $\{B\}$ at $o$ and a frame $\{E\}$ to the body at some point $o'$ of
the body as shown in Fig.\ref{frames_se3}.  At each instance the configuration of the rigid
body is described via a transformation matrix, $A\in SE(3)$, corresponding to
the displacement from frame $\{B\}$ to frame $\{E\}$.

\begin{figure}[h]
	\begin{center}
		\includegraphics[scale=.3]{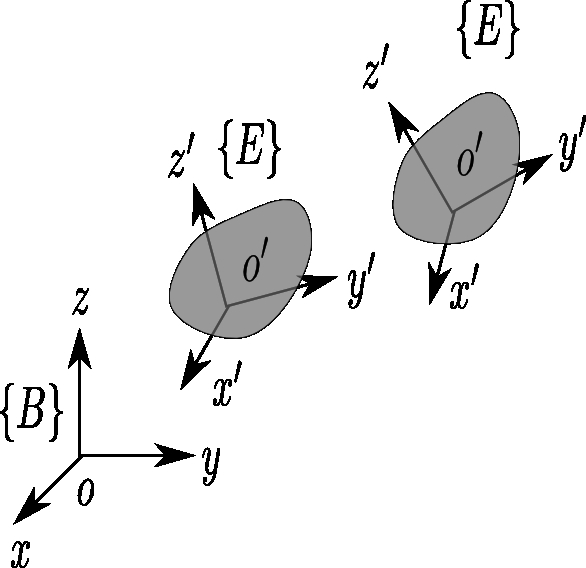}
	\end{center}
	\caption{Inertial frame $\{B\}$ and body fixed frames $\{E\}$}
	\label{frames_se3}
\end{figure}

So a rigid body motion becomes a curve in $SE(3)$, let $A(t)$ be such a curve
given by $A(t) : [-c,c]\to SE(3), A(t)=\left[\begin{array}{cc} R(t) & d(t)\\0 & 1\end{array}\right]$.
The lie algebra element $S(t)\in se(3)$  can be identified with the tangent
vector $A'(t)$ at an arbitrary t by:
\begin{equation}
S(t) = A^{-1}(t)A'(t) = \left[\begin{array}{cc} [\omega](t) & v(t)\\ 0 & 0 \end{array}\right]
\end{equation}
The $\omega$ physically corresponds to the angular velocity of the body, while
$v$ is the linear velocity of the origin $O'$.  
Let us assign a riemannian metric 
$g = \left[ \begin{array}{cc} \alpha I_3 & 0\\ 0 & \beta I_3\end{array}\right]$
over $SE(3)$ as prescribed in \cite{park95}.  And so for $V = (\omega,v)\in se(3)$,
$\langle V,V \rangle_g = \alpha \omega^T \omega + \beta v^T V$.
It was proved in \cite{zefran95} that the analytic expression for the geodesic between 
two configurations $A_1$ and $A_2$ in $SE(3)$,
with $g$ as riemannian metric,  is given by;
\begin{eqnarray}
R(t) = R_1 \exp([\omega_0]t)\\
d(t) = (d_2-d_1)t+d_1
\end{eqnarray}
where $[\omega_0]=log(R_1^T R_2)$ and $t\in [0,1]$.  
The path is unique for $Trace(R_1^T R_2)\neq -1$.
And the distance between two configuration in $SE(3)$ is given by 
\begin{equation}
d(A_1,A_2)=\sqrt{\alpha\Vert \log(R_1^{-1} R_2) \Vert^2+\beta\Vert d_2-d_1 \Vert^2}.
\label{metric_se3}
\end{equation}
All the formulas required for computing $\exp$ and $\log$ maps are given in the 
Appendix.\ref{ap1} for completeness.
\begin{exm}
Consider two configurations $A_1$ and $A_2$, as
shown in Fig.\ref{exp_se3_exp2}, given by
vectors $(\omega_1,v_1)$ and $(\omega_2,v_2)$ respectivel, where
$\omega_1=\frac{\pi}{4}\left[\begin{array}{ccc}1 & 0 & 0\end{array}\right]$,
$v_1=\left[\begin{array}{ccc}-6 & 0 & 0\end{array}\right]$,
$\omega_2=\frac{\pi}{2}\left[\begin{array}{ccc}1 & 1 & 0\end{array}\right]$
and
$v_2=\left[\begin{array}{ccc} 0 & 6 & 2\end{array}\right]$.
\begin{figure}[h]
	\begin{center}
		\includegraphics[scale=.5]{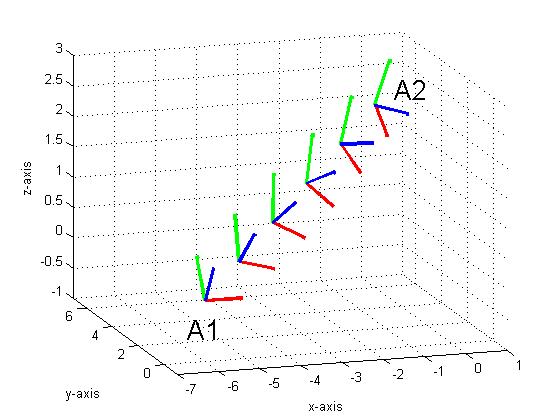}
	\end{center}
	\caption{A geodesic between $A_1,A_2\in SE(3)$.}
	\label{exp_se3_exp2}
\end{figure}
\end{exm}
$SE(3)$ is used extensively in robotics for path planning
and motion planning of robots.  It is also useful in 
computer vision and graphics.

Once the Riemannian metric is identified we can construct a distance metric on the manifold.  
With the distance metric $d(\cdot,\cdot)$(corresponding to the geodesic path) 
defined on the Riemannian manifold we are now ready to
talk about the medial axis and the sampling criterion for a curve on the manifold. 

\section{Medial Axis, Dense sample and Flatness}
	\label{medial}
We proceed by revisiting the definition of the medial axis stated
previously.  
Let $\mathcal{M}$ be a Riemannian manifold and
$d(\cdot,\cdot):\mathcal{M}\times \mathcal{M}\to \mathbb{R}$
be the corresponding distance metric.
\begin{defi}
The medial axis $M$ of a curve $\mathcal{C}\subset \mathcal{M}$,  
is the closure
of the set of points in $\mathcal{M}$ that have at least two closest
points in $\mathcal{C}$.  
\end{defi}
Fig.\ref{medial_ellipse} shows examples of
medial axis of closed curves on a half cylinder and in a plane.
\begin{figure}[h]
	\begin{center}
		\includegraphics[scale=.22]{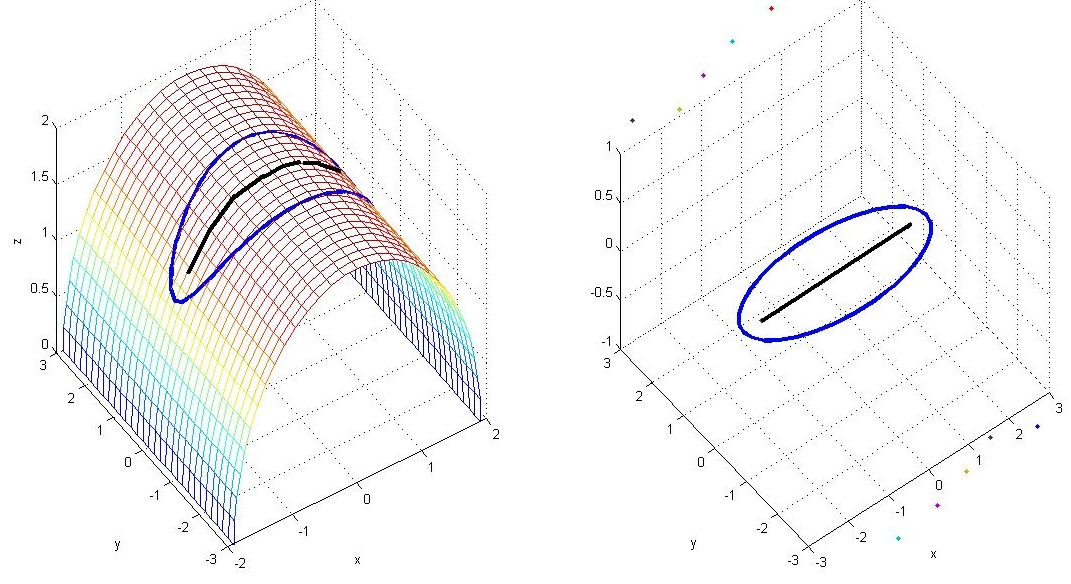}
	\end{center}
	\caption{Medial axis of a curve on a surface and a curve in a plane}
	\label{medial_ellipse}
\end{figure}
It should be noted that the medial axis, as defined above,
is a subset of the underlying manifold in which the curve lies.
A curve embedded in a riemannian manifold and embedded in $\mathbb{R}^3$
will have different medial axes.
The open disc (ball) 
of radius $\varepsilon > 0$ in $\mathcal{M}$ with $s\in \mathcal{M}$ as a center 
is defined as $S_{\varepsilon}(s)=\{x\in\mathcal{M}\vert d(s,x)<\varepsilon\}$.
In the same manner $B_{\varepsilon}(s)=\{x\in\mathcal{M}\vert d(s,x)\leq \varepsilon\}$
is a closed disc (ball) in $\mathcal{M}$ with radius $\varepsilon$ and the center $s$.
The set $\partial B_{\varepsilon}(s) = \{ x\in \mathcal{M} \vert d(x,s) = \varepsilon \}$
is the boundary of $B_{\varepsilon}(s)$.

\begin{defi}
At a point $p$ on the curve $\mathcal{C}$ the local feature size $f(p)=d(p,M)$.  Where
$d(p,M)=\inf \{d(p,m),\forall m\in M\}$.  
\end{defi}
The local feature size at a point on the curve captures the behaviour of the
curve in the neighbourhood of that point.  
In practice for arbitrary curves it is difficult to identify the 
medial axis.  Looking at the construction of the \textit{Voronoi} diagram\cite{rourke98} 
for a given sample points of a curve the
Voronoi vertices do capture the behaviour of the medial axis of the
sampled curve.
And so for a densely sampled curve the Voronoi vertices for these samples
are taken to be the approximate of the medial axis of the given curve.
It is computationally challenging to construct voronoi diagrams on curved
spaces \cite{leibon00}.
\subsection{Dense sampling}		
A tubular neighbourhood for a curve in a plane is defined as the subset of the 
plane such that every point of the subset belongs to exactly one line
segment totally contained in the subset and normal to the curve.  And a disk
centred on the curve contained in a tubular neighbourhood of the curve is called a 
tubular disk.  Let us generalize this definition to curves in manifolds.  We will
also define the notion of a dense sample of a curve in manifold based on 
the tubular neighbourhood.
\begin{defi}
Let $\mathcal{C}\subset \mathcal{M}$ be a smooth curve.  Consider segments of geodesics 
that are normal to $\mathcal{C}$ and start in $\mathcal{C}$. If  $\mathcal{C}$ is compact, 
then there exists an $\varepsilon>0$ such that no two segments of length  $\varepsilon$
and starting at different points of $\mathcal{C}$ intersect \cite{spivak1}. 
The union of all such segments 
of length $\varepsilon$ is an open neighbourhood $T$ of $\mathcal{C}$, and is called a 
tubular neighbourhood of $\mathcal{C}$.
\end{defi}

We denote the open segment with center $p\in\mathcal{C}$ and
radius $\varepsilon$ in the normal geodesic segment of $\mathcal{C}$ at $p$ by $N_{\varepsilon}(p)$.
Revisiting the definition of the tubular neighbourhood: the union 
$N_{\varepsilon}(\mathcal{C})=\cup_{p\in\mathcal{C}} N_{\varepsilon}(p)$ is called
a tubular neighbourhood of radius $\varepsilon$ if it is open as a subset
of $\mathcal{M}$ and the map 
$F:\mathcal{C}\times (-\varepsilon,\varepsilon)\to N_{\varepsilon}(\mathcal{C})$ is a 
diffeomorphism.  This interpretation is the outcome of result from \cite{carmo92}.
Let $\mathcal{C}\subset \mathbb{R}^2$, be a simple closed smooth curve.
Existence of the tubular neighbourhood is evident from the compactness of the
curve in $\mathbb{R}^2$.  We show something more about the value of 
$\varepsilon$ in next proposition.

\begin{prp}
If $N_{\varepsilon}(\mathcal{C})$ is a tubular neighbourhood of $\mathcal{C}$ then 
$\varepsilon < \frac{1}{k}$.  Where $k = \max{k(p)}, p\in\mathcal{C}$ and 
$k(p)$ is the curvature of the curve at point $p$.
\end{prp}
\begin{proof}
Let us define a curve $\alpha(s)$ in $\mathbb{R}^2$ by
\[
\alpha(s)=F(\mathcal{C}(s),t)=\mathcal{C}(s)+t N(\mathcal{C}(s)),
\]
for a fixed $t\in(-\varepsilon,\varepsilon)$ such that $\alpha(0)=p$. $N(p)$ 
is the unit normal to the curve $\mathcal{C}$
at $p$.  This new curve belongs to the open set $N_{\varepsilon}(\mathcal{C})$
and 
\begin{eqnarray}
\alpha(0) = p+t N(p) \\
\alpha'(0) = \mathcal{C}'(0)+t \mathcal{C}'''(0) \\
\alpha'(0) = (1 - t k(p)) \mathcal{C}'(0) = (dF)_{(p,t)}(\mathcal{C}'(0))
\end{eqnarray}
Since $F : \mathcal{C} \times (-\varepsilon,\varepsilon) \to \mathbb{R}^2$
is a diffeomorphism when restricted to $\mathcal{C}\times(-\varepsilon,\varepsilon)$,
we have that $(dF)_{(p,t)}(\mathcal{C}'(0))$ is a non-null vector,i.e. 
$1-t k(p) \neq 0$.  But $(-\varepsilon,\varepsilon)$ is connected and 
$1-t k(p)>0$  for $t=0$.  Thus $1-t k(p)>0$ on 
$\mathcal{C}\times (-\varepsilon,\varepsilon)$.  Now if we find out 
$k = \max k(p), p\in \mathcal{C}$ then $1- t k > 0 $.
And we have $\varepsilon = t < \frac{1}{k}$.
\end{proof}

\begin{defi}
A finite sample set $\mathcal{S}\subset \mathcal{C}$ is called a \textit{uniform
$\varepsilon-$sample} for some $\varepsilon>0$ if for any two consecutive sample 
points $r,s\in \mathcal{S}$, $r \in B_{\varepsilon}(s)$. 
\end{defi}

\begin{defi}
A uniform $\varepsilon$-sample $\mathcal{S}$ of a curve $\mathcal{C}\subset\mathcal{M}$ 
is dense if there is a real number
$\varepsilon>0$ such that $\cup_{s\in\mathcal{S}} B_{\varepsilon}(s)$,i.e. the union of 
the closed disks of radius $\varepsilon$ centred at the sample points $s\in\mathcal{S}$, forms a 
tubular neighbourhood of $C$.
\end{defi}

\begin{prp}
For plane curves if $\varepsilon < \min_{p\in \mathcal{C}}{f(p)}$ then
the uniform $\varepsilon$-sample $\mathcal{S}$ of curve $\mathcal{C}$ is a dense sample.
\end{prp}
\begin{proof}
By the definition of $f(p)$.
\end{proof}

Before we proceed to the main theorem we will discuss few observations in the next
section.  We show by an example how the medial axis based sampling 
fails due to the curvature of the underlying riemannian manifold.
We also propose to work within the injectivity radius of the manifold
to avoid such a problem.
\subsection{Observations and a Counter Example}
The first two observations are encouraging.  And the 
counter example to these to in next section helps in
identifying a conservative sampling condition.  We know 
that to form a dense sample of a curve in $\mathbb{R}^n$ it
is required to sample with the $\varepsilon_1 < \min_{p\in\mathcal{C}}f(p)$.
The curve and corresponding $\varepsilon_1$ is shown in Fig.\ref{obs1_1}.  
\begin{figure}[ht]
\centering
\subfigure[]{
\includegraphics[scale=.25]{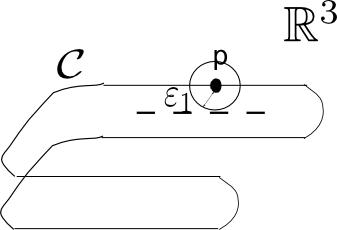}
\label{obs1_1}
}
\subfigure[]{
\includegraphics[scale=.25]{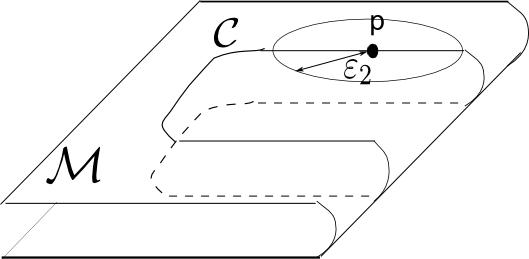}
\label{obs1_2}
}
\caption[obs1]{
	\subref{obs1_1} 
	A curve $\mathcal{C}\in\mathbb{R}^3$ and the 
	part of medial axis near $p\in \mathcal{C}$.
	$\varepsilon_1$  is the distance of the point
	$p\in \mathcal{C}$ from the medial axis of the
	curve in space.
	\subref{obs1_2} 		
	The same curve $\mathcal{C}$ on a surface $\mathcal{M}$
	and the medial axis distance $\varepsilon_2$ from 
	the point $p\in\mathcal{C}$ to the medial axis of the
	curve on the surface.}
\label{observation1}
\end{figure}
Whereas if the same curve
is embedded in a surface, as shown Fig.\ref{obs1_2}, the required
$\varepsilon_2$ needs to be evaluated on the surface.  It turns out that
$\varepsilon_2<\varepsilon_1$.  Let us look at another example.
\begin{figure}[ht]
\centering
\subfigure[]{
\includegraphics[scale=.3]{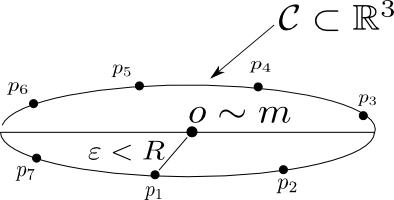}
\label{obs2_1}
}
\subfigure[]{
\includegraphics[scale=.3]{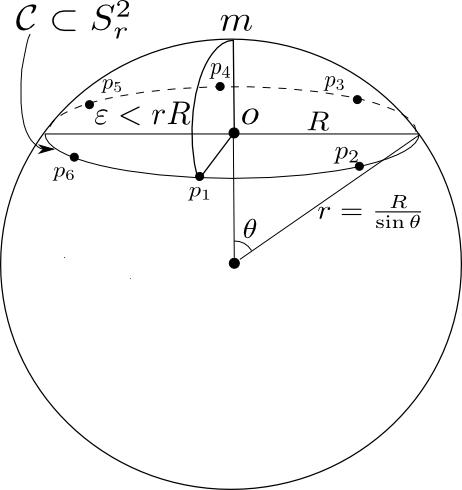}
\label{obs2_2}
}
\subfigure[]{
\includegraphics[scale=.3]{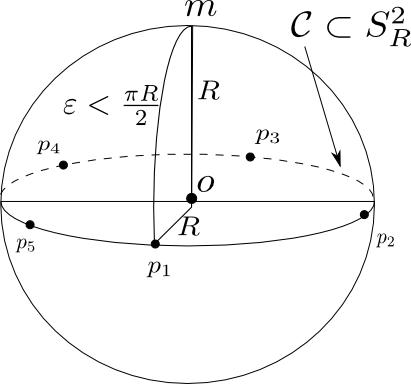}
\label{obs2_3}
}
\caption[obs2]{
	\subref{obs2_1} Circle with radius R is lying in space 
	\subref{obs2_2} Circle resting on a sphere of radius $r=\frac{R}{\sin \theta}$		
	\subref{obs2_3} Circle is the great circle the sphere of radius $R$}
\label{observation2}
\end{figure}
A circle in a $xy$-plane in $\mathbb{R}^3$ can be 
thought of as some latitude on a sphere of radius 
$r\geq \frac{L}{2\pi}$, where $L$ is the length of the
circle.  In both the cases,
circle on a plane and circle on a sphere,
the sampling required for correct reconstruction is different.
On sphere we need less dense sample set as compared to
the plane.  In fact as we increase the radius $r$ we
need denser and denser sample set for correct reconstruction
and its limiting case, $r\to \infty$, is the plane.  
In $\mathbb{R}^3$ the usual 
euclidean metric is carried over to the points of
the circle.  In case of sphere the shortest path
between two points is always along the great
circle passing through these two points and the 
length of the segment which is shorter is the distance
between two points on sphere.  With this distance
metric defined,
sphere becomes a metric space.  And the points of
the circle on sphere are endowed with this metric.  
The points of this circle
on a sphere are more structured then the points of
the same circle in the space.  The additional knowledge of 
the underlying surface adds up to the ordering
relation between points of the circle.  Since we know
the surface we know the tangent space and that reduces
the effort in ordering the sample points.

Interestingly, when generalized to the curves on manifolds, the sampling criterion based
only on the medial axis becomes meaningless. 
As an example let us look at a circle of radius one on the surface shown
in Fig.\ref{medial_violate}.
\begin{figure}[h]
	\begin{center}
	\includegraphics[scale=.45]{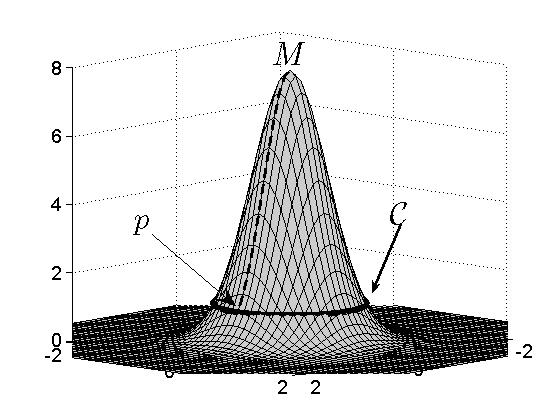}
	\end{center}
	\caption{A circle $\mathcal{C}$ and the normal geodesic from a point $(1,0,1.0629)$ to
	$M$}
	\label{medial_violate}
\end{figure}
The medial axis of the circle on the given surface is the point $M=(0,0,0)$.  For any point 
on the circle, the distance from the medial axis turns out to be larger than the length of the
circle itself. 
Consider the limiting case of this surface, a cylinder, suppose the circle is on the cylinder.  The 
medial axis point does not exist.

The above phenomenon can be understood clearly if we look at the cut locus of the point
$p\in\mathcal{M}$.  Following can be considered as the defining property of
the cut locus of a point on the manifold.
If $\gamma (t_0)$ is the cut point of $p=\gamma (0)$ along
the geodesic arc $\gamma$ then either $\gamma (t_0)$ is the first conjugate point
of $\gamma (t_0)$ or there exists a geodesic $\sigma = \gamma$ from 
$p$ to $\gamma (t_0)$ such that $l(\sigma)=l(\gamma)$(lengths of
$\sigma$ and $\gamma$ are equal).  

For example if $\mathcal{M}$ is a sphere $S^2$ and $p\in S^2$ then the
cut locus of $p$ is its antipodal point.  And if we consider the sphere
of radius $R$ the distance of point $p$ from its cut locus is
$\pi R$.  Whereas the distance of the point $p$ 
on the circle in Fig.\ref{obs2_3} to the medial axis $M$
is $\frac{\pi R}{2}$.
Now coming back to the counter example Fig.\ref{medial_violate} we observe that 
the distance of the $p$ to its cut locus $d(p,C_m(p))$
is less then the distance to the medial axis $M$ of the circle.  Where
$C_m(p)$ is the cut locus of $p\in\mathcal{M}$. 

It can be shown that if $q\in\mathcal{M}-C_m(p)$ there exists
a unique minimizing geodesic joining $p$ and $q$.  In 
\cite{carmo92} 
\begin{equation}
i(\mathcal{M}) = \inf_{p\in\mathcal{M}} d(p,C_m(p))
\end{equation}
is called the injectivity radius of $\mathcal{M}$.
So if $\varepsilon<i(\mathcal{M})$ then $\exp_p$ is
injective on the open ball $S_\varepsilon(p)$.

Tubular neighbourhood for a curve is constructed by
taking only the normal geodesics to the curve at a point
and assuring the injectivity of the $exp_p$ map along
these normal directions.  We now propose to
work inside the injectivity radius to straighten
out the problem with sampling.

\begin{prp}
Let $\mathcal{C}\in\mathcal{M}$ be a smooth, simple
and closed curve.  If $\mathcal{S}$ is a uniform
$\varepsilon$- sample of $\mathcal{C}$ then
$\mathcal{S}$ is dense for  
$\varepsilon < \min \{\inf_{p\in\mathcal{C}}f(p),i(\mathcal{M})\}$.
\label{epsilon}
\end{prp}

\subsection{Flatness of the curve inside a tubular neighbourhood}
If the underlying manifold is a plane and a curve is sampled densely then based
on the tubular neighbourhood it is proven, in \cite{gomes94}, that
euclidean minimal spanning tree reconstructs the sampled arc.
The crucial argument for the correctness of above is the denseness of the sample.
It comes from the observation that an arc does not wander too much inside a tubular disk.
Which avoids the connections between the non-consecutive sample points in $\mathcal{S}$,
defined as short chords. 

Now we give an alternate proof of flatness of the curve segment inside a tubular neighbourhood
in plane.  And after that we extend the proof to curves in the riemannian manifold.

\begin{thm}
Let $p$ and $q$ be two points on an arc $\mathcal{C}\subset\mathbb{R}^2$ such that $q$ is inside the 
tubular disk $B_{\varepsilon}(p)$ centred at $p$. Then the sub arc $pq$ of $\mathcal{C}$ 
is completely inside $B_{pq/2}(c)$, where $c$ is the center of
diameter $pq$.
\label{alt_plane}
\end{thm}
\begin{proof}
Since $q\in B_{\varepsilon}(p)$, $pq=d(p,q)\leq \varepsilon$.
Now pq being a segment of an arc $\mathcal{C}$
there are three possible ways, as shown in Figure.\ref{plane_alternate_proof}, in which 
it intersects with $B_{pq/2}(c)$. 
\begin{figure}[ht]
\centering
\subfigure[]{
\includegraphics[scale=.25]{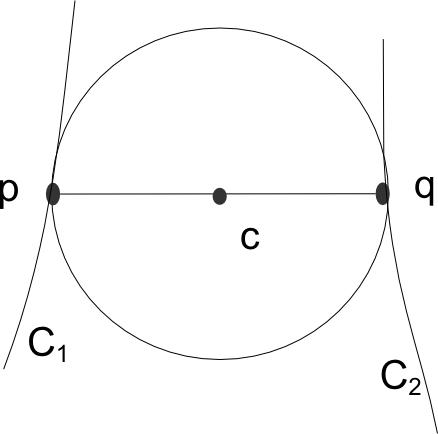}
\label{case1}
}
\subfigure[]{
\includegraphics[scale=.25]{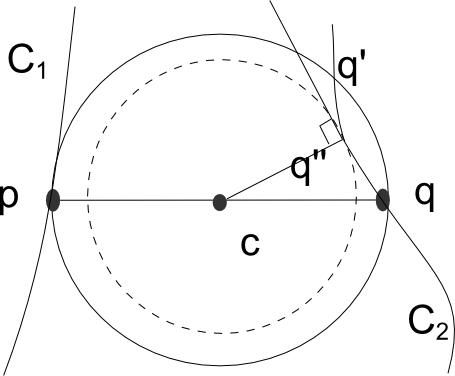}
\label{case2}
}
\subfigure[]{
\includegraphics[scale=.25]{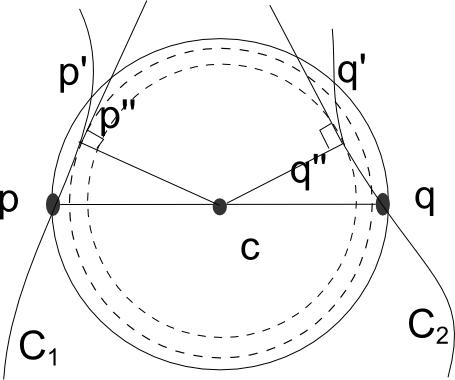}
\label{case3}
}
\caption[Alternate proof]{
	\subref{case1} 
	The arc touches $B_{pq/2}(c)$
	\subref{case2} 		
	The arc touches $B_{pq/2}(c)$ at $p$ and 
	intersects its boundary at $q'$ while passing through $q$ 
	\subref{case3}
	The arc intersects boundary of $B_{pq/2}(c)$ at $p'$ 
	and $q'$ while passing through $p$ and $q$ respectively}
\label{plane_alternate_proof}
\end{figure}

For the possibility shown in 
Figure.\ref{case1}, it is evident that
center $c$ lies on two normals passing through $p$ and $q$, i.e.
$c\in\overline{pq}$,
since $B_{pq/2}(c)$ and $\mathcal{C}$ share common tangents
at $p$ and $q$.  This can not happen
since $B_{pq/2}(c)\subset B_{\varepsilon}(p)$, a subset of a tubular
neighbourhood.  

Let us consider the case in Figure.\ref{case2},
arc $\mathcal{C}$ touches $B_{pq/2}(c)$ at
$p$ and intersects the boundary of $B_{pq/2}(c)$ at $q$ and $q'$.
We can find out a point $q''$ on the segment $qq'$ which is nearest to $c$.
At $q''$ the circle with center $c$ and radius $d(c,q'')$ shares a common
tangent with $\mathcal{C}$.  Hence $c$ lies on the two normals 
$\overline{pc}$ and $\overline{q''c}$.  This can not happen inside
a tubular neighbourhood.

Finally we consider the Figure.\ref{case3}.  On segments
$pp'$ and $qq'$ we find $p''$ and $q''$ nearest to $c$.
In this case $c$ lies on $\overline{p''c}$  and 
$\overline{q''c}$.  Since $c$ is inside tubular neighbourhood
this can not happen.

So the only possibility we are left with is that the segment
$pq$ of curve $\mathcal{C}$ lies entirely inside $B_{pq/2}(c)$.
\end{proof}
\begin{thm}
Let $p$ and $q$ be two points on an arc $\mathcal{C}\subset\mathcal{M}$, where 
$\mathcal{M}$ is any Riemannian manifold, such that $q$ is inside the 
tubular disk $B_{\varepsilon}(p)$ centred at $p$. Then the sub arc $pq$ of $\mathcal{C}$ 
is completely inside $B_{pq/2}(c)$, where $c$ is the center of
diameter $pq$.
\label{alt_surf}
\end{thm}
\begin{proof}
For $\mathcal{M}:=\mathbb{R}^n$ we know that $\overline{cp}$, $p\in S^{n-1}$,
is orthogonal to $T_pS^{n-1}$. 

Since we are working inside a tubular neighbourhood of the curve $\mathcal{C}$,
with $\varepsilon$ as prescribed in Proposition.\ref{epsilon}, $\exp : T_p\mathcal{M}\to\mathcal{M}$
is a diffeomorphism.
\begin{figure}
\begin{center}
	\includegraphics[scale=.25]{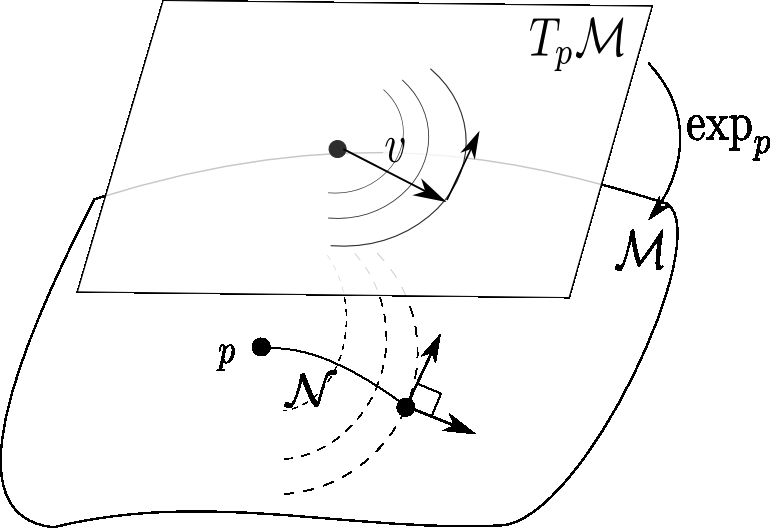}
	\caption{Tangent space of a point $p\in\mathcal{M}$ where
	$||v||<\varepsilon$ and the corresponding geodesic
	$\mathcal{N}$.}
\label{gausslm}
\end{center}
\end{figure}
Gauss's lemma \cite{carmo92}
asserts that the image of a sphere of sufficiently small radius($<\varepsilon$) 
$T_p\mathcal{M}$ under the exponential map is perpendicular to all
geodesics originating at $p$ Fig.\ref{gausslm}.

And the rest of the proof follows from arguments of 
the Theorem.\ref{alt_plane}.
\end{proof}

\section{Curve Reconstruction on a Riemannian manifold}
\label{order}
\subsection{Ordering}
We model a curve with a graph where the 
vertices of the graph are the sample points and the edges
indicate the order in which the vertices are connected.
This also implies a geometric realization of the graph.
If further we put the distance between two sample points
as the edge cost, it becomes a weighted graph.
A minimal spanning tree for a weighted graph is a spanning tree for which
the sum of edge weights is minimal.  
To keep the notations consistent we define the geodesic polygonal path
on riemannian manifold as the path along which every 
vertex(sample point) pair is connected by a geodesic segment.

Computing the minimal spanning tree use the following fundamental property,
let $X\cup Y$ be a partition of the set of vertices of a connected weighted
graph $G$.  Then any shortest edge in $G$ connecting a vertex of $X$ and a
vertex of $Y$ is an edge of a minimal spanning tree.
If we use MST to model an arc, we must ensure that there are no short
chords in the graph, proved in \cite{gomes94}.  

Our work focuses on closed, simple, smooth curves. We expect
the MST for which every vertex has degree two.
In other words the sample point has exactly two neighbours 
(samples) on the curve.

\begin{thm}
If $\mathcal{S}$ is
a dense sample of $\mathcal{C}\subset\mathcal{M}$ then MST gives a correct
geodesic polygonal reconstruction of $\mathcal{S}$.  Where $\mathcal{C}$
is a smooth, closed and simple curve.
\end{thm}
\begin{proof}
We show that the geodesic polygonal path has no short chords.  The argument
is similar to the proof provided for planar case in \cite{gomes94}.  For
the completeness of the article we restate the argument here.
Suppose that MST does not give a correct geodesic polygonal reconstruction
of $\mathcal{S}$.  It implies that there are two points in MST which
are not consecutive.  Let these points be $p,q \in \mathcal{S}$.  
Since $pq$ is a short chord there has to be at least one edge 
in the sub arc $pq$ which has length greater than that of $pq$.
But since the sample $\mathcal{S}$ is dense, the arc $pq$ must be
contained in the disc with diameter $pq$.  Inside the disc 
there is no arc with length greater than the length of the
diameter.  So we have a contradiction.
\end{proof}

\subsection{Interpolation}
Once we have ordered the given set of points of the curve on a
curved manifold the next step is to interpolate this point
set to the desirable granularity.  The easiest way to interpolate
the points is to connect the points via straight line segments,
a linear interpolation.  In general for a manifold like 
$SE(3)$, the geodesics are the $\exp$ segments.  But
this scheme will not produce a differentiable curve which
might be necessary for some applications.  Based on the need
and application one may chose the interpolation scheme. 
In \cite{sho85} and \cite{kim95} a quaternion based approach 
is suggested and is very useful in computer graphics and animation.
Since we have represented $SE(3)$ using matrices we would
rather stick to matrices.  Motivated by motion
planning purposes various interpolation schemes
based on variational minimization techniques have been 
proposed and some of them turn out to be quite simple
for implementations, for a broad overview one will find
\cite{parkravani97} and \cite{li2006}
useful.  For the completeness of the reconstruction
process we have used \textit{de Casteljau} construction
as prescribed in \cite{alt2000}, i.e.
generalizing the multilinear interpolation on 
$SE(3)$, a piecewise $C^2$ curve connecting
two frames with given velocities.  The advantage is
that the expression is in the closed form with exponential and
log maps. 
\begin{figure}[h]
	\begin{center}
		\includegraphics[scale=.25]{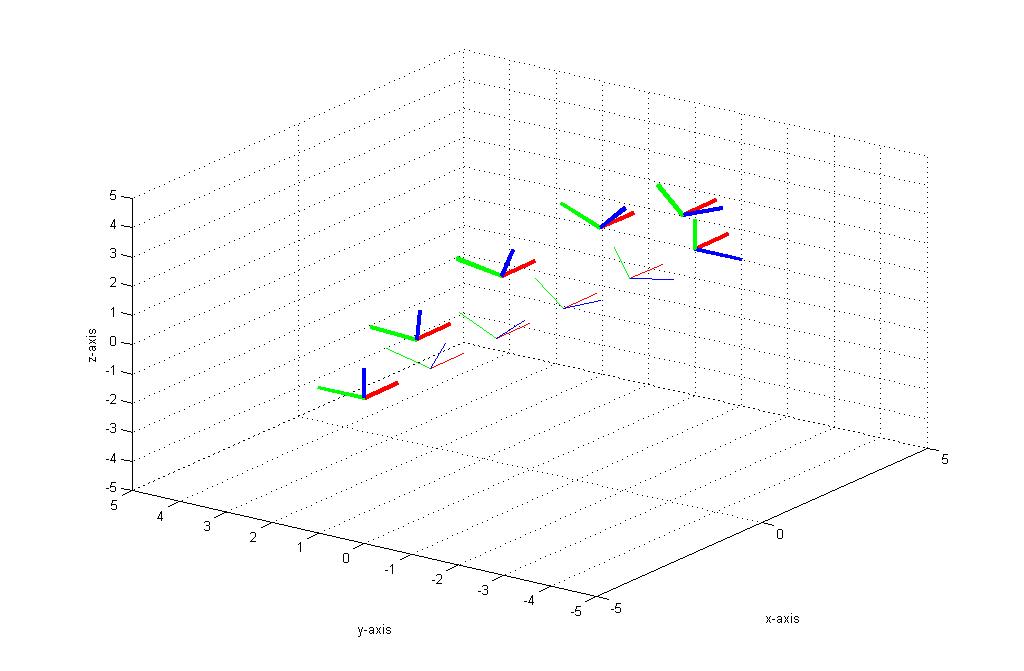}
	\end{center}
	\caption{Comparision of Exponential map and $C^2$ smooth interpolation
	in $SE(3)$ between $g_0=[0,0,0]\ltimes[-5,0,0]$ and $g_1=[\pi/2,0,0]\ltimes[5,0,0]$,
	with tangents $v_0^1=[0,0,0,3,1,1]$ and $v_2^1=[\pi/2,0,0,-1,-3,-1]$.}
	\label{comp_exp_decast}
\end{figure}
Suppose we do not know the velocities at the node points.
For such a case we have used a partial geodesic scheme to interpolate
between two elements of $SE(3)$.  Where, the rotational part is
interpolated by the $\exp$ map and the translational
component is a interpolated with spline segments.

\section{Simulations}
\label{sim}
\subsection{Curves on a Sphere}
We begin our simulations with examples of curves on a unit
sphere.  We show two curves with different densities required
by the MST for correct reordering of the samples.  
\begin{figure}[ht]
\centering
\subfigure[]{
\includegraphics[scale=.4]{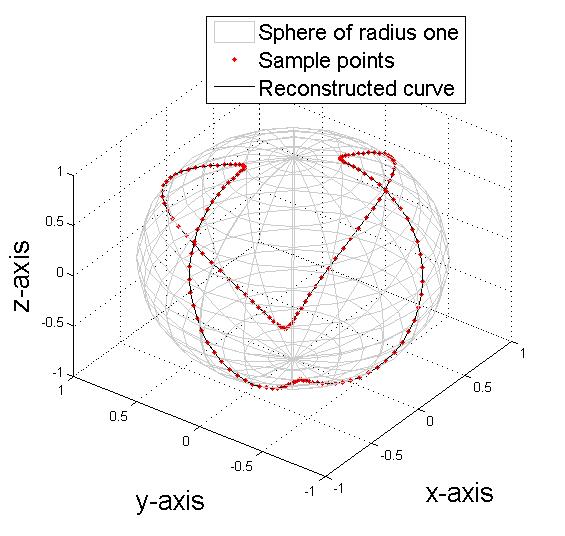}
\label{curve_sphere1}
}
\subfigure[]{
\includegraphics[scale=.4]{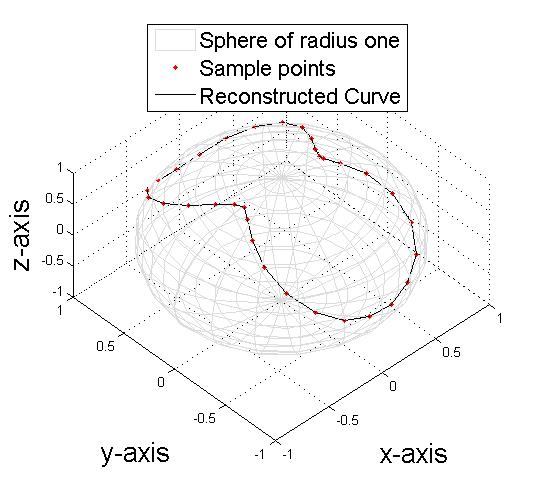}
\label{curve_sphere2}
}
\caption[curve_sphere]{
Example curves on a unit sphere.
}
\label{curve_sphere}
\end{figure}

The curves after reordering the sample points are 
shown in Fig.\ref{curve_sphere1} and Fig.\ref{curve_sphere2}.

\subsection{Curves in $SE(3)$}
In Fig.\ref{unordered_se3_1} an unordered set of frames in 
$SE(3)$ are shown.  We assume that the sample shown
is dense.

\begin{figure}[h]
	\begin{center}
		\includegraphics[scale=.2]{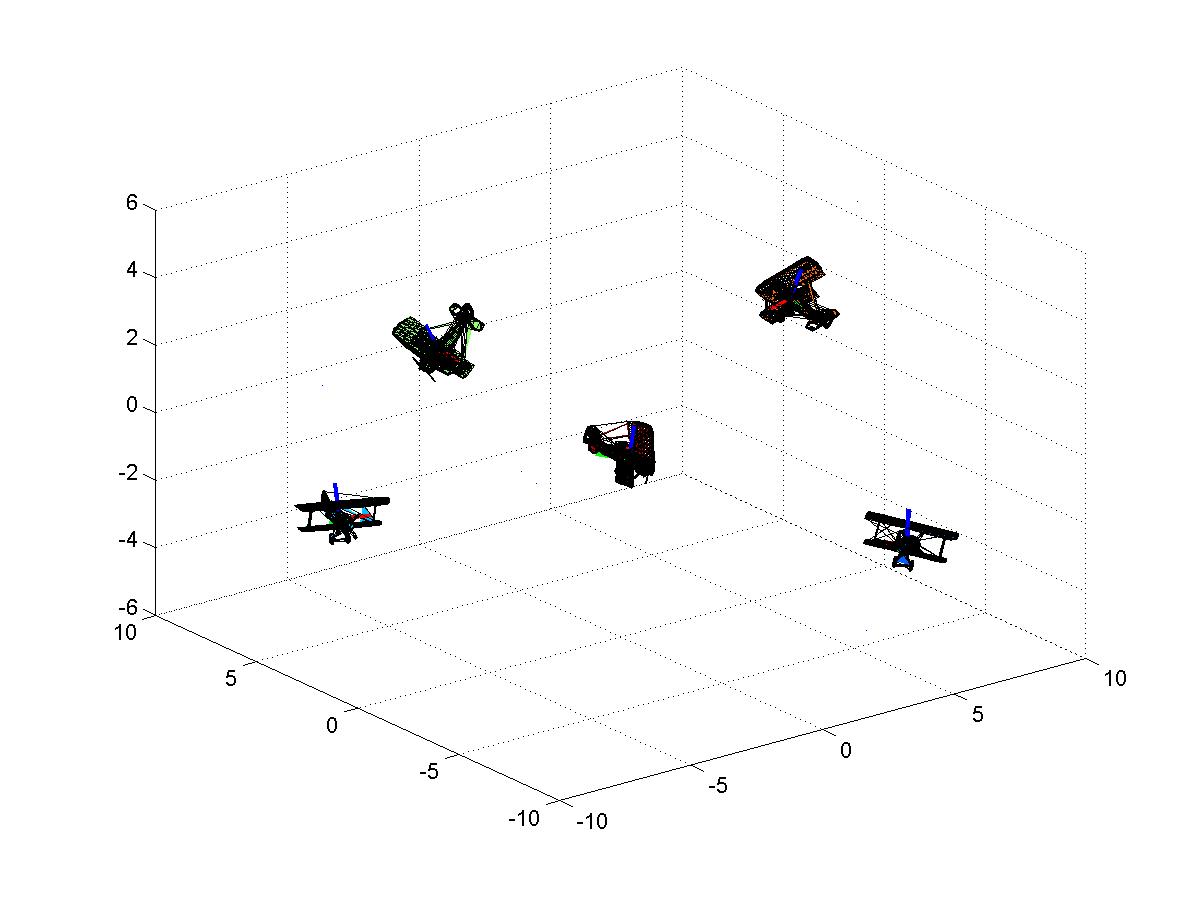}
	\end{center}
	\caption{A sample $\mathcal{S}$ of a curve $\mathcal{C}\subset SE(3)$}
	\label{unordered_se3_1}
\end{figure}

By the distance metric defined in Eqn.\ref{metric_se3} we compute
distances between all the frames.  Finally we compute the MST for the
complete weighted graph of frames with the computed distances as the
edge weights.  

\begin{figure}[h]
	\begin{center}
		\includegraphics[scale=.2]{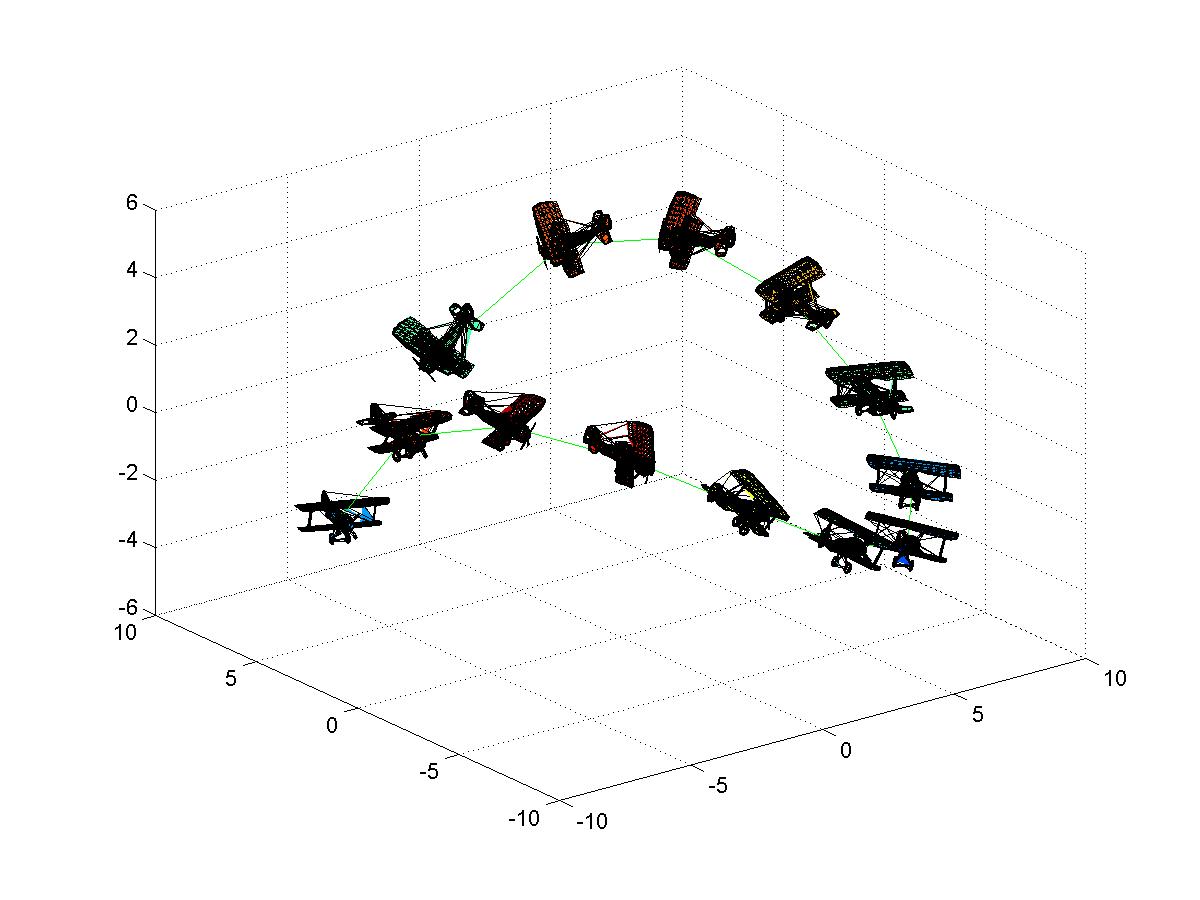}
	\end{center}
	\caption{Reconstructed curve in $SE(3)$}
	\label{se3_1}
\end{figure}

Once the ordering is done we interpolate the sample with partial geodesic
scheme.  Results of interpolation with two different granularities is
presented in Fig.\ref{se3_1} and Fig.\ref{se3_2}.

\begin{figure}[h]
	\begin{center}
		\includegraphics[scale=.2]{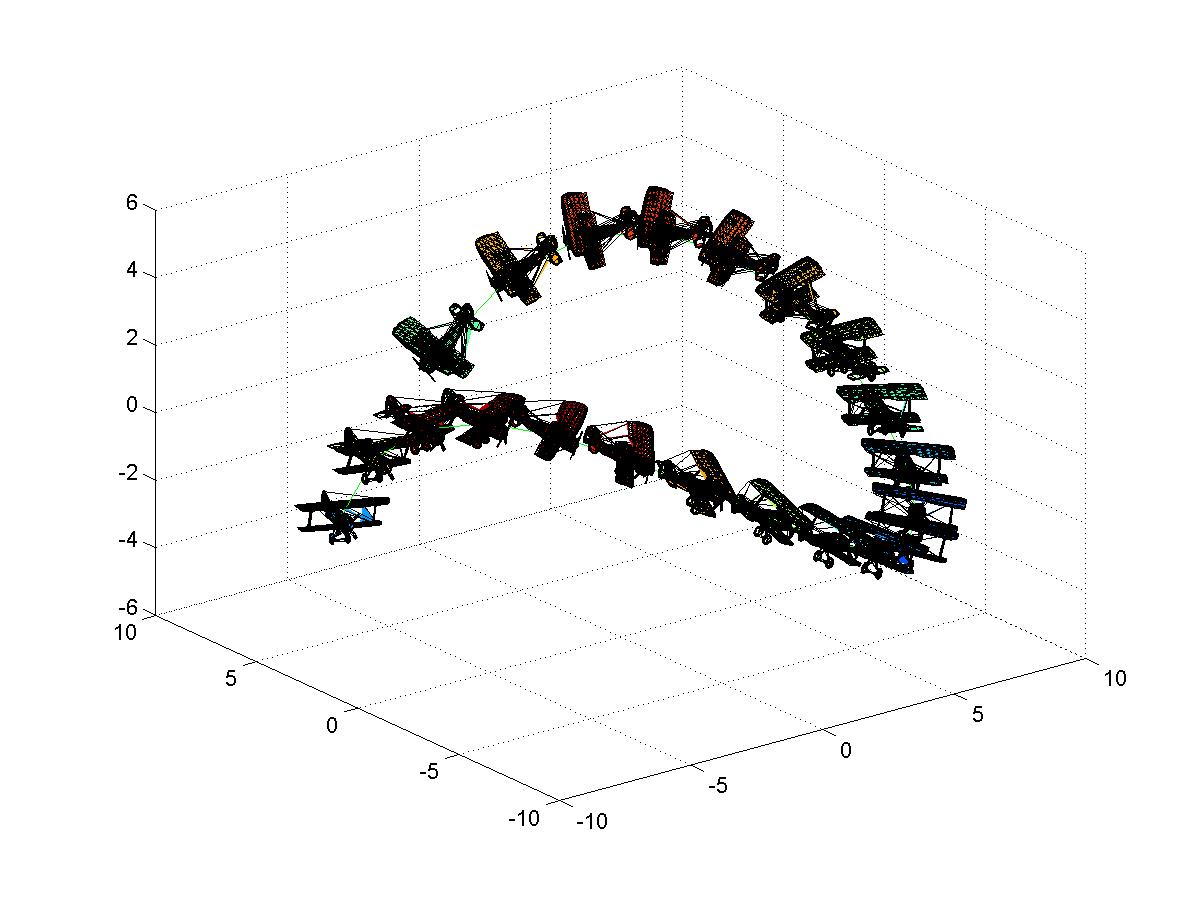}
	\end{center}
	\caption{Reconstructed curve in $SE(3)$ with finer interpolation}
	\label{se3_2}
\end{figure}

\subsection{Curve in $SE(3)$ with scaling parameter}
Suppose for a planar object in motion, we include 
scaling with respect to the center of 
mass along with the rotation and 
translation.  The resultant element will be of the
following form
\begin{equation}
A = \left[\begin{array}{cc}
e^{\lambda} R & d\\
0 & 1
\end{array}\right].
\label{ese2_form}
\end{equation}
This element operates on the point of the object in plane.
It scales($e^{\lambda}$) and rotates($R$) the object with respect to its center
of mass and then translates($d$) the center of mass.
With each such element we can associate a vector
$[\lambda, \theta, d_x, d_y]$.  The elements of the form given by
Eq.\ref{ese2_form}
with standard matrix multiplication forms a lie group.
We can extend the notions of tangent space and
exponential map to this lie group.  As discussed previously in
Sec.\ref{metric_emg} this group is a semi-direct product of 
elements of scaled rotations and translations.
The tangent space elements at identity, lie algebra elements, 
for scaled rotations are given by
\begin{equation}
[a] =
\lambda \left[\begin{array}{cc} 1 & 0 \\ 0 & 1\end{array}\right]
+
\theta \left[\begin{array}{cc} 0 & -1 \\ 1 & 0\end{array}\right]
\end{equation}
And the usual matrix exponentiation gives 
\begin{equation}
\exp{[a]} = e^{\lambda}
\left[\begin{array}{cc}
\cos{\theta} & -\sin{\theta}\\
\sin{\theta} & \cos{\theta}
\end{array}\right].
\end{equation}

\begin{figure}[ht]
\centering
\subfigure[]{
\includegraphics[scale=.2]{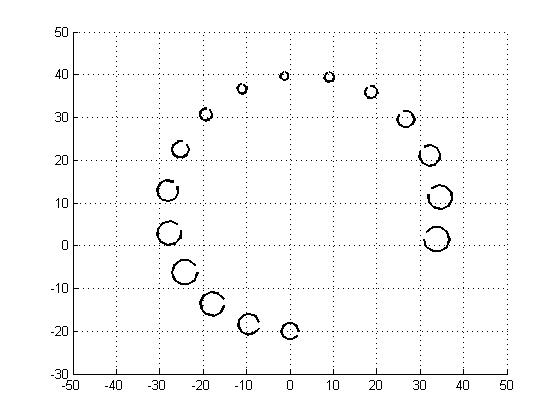}
\label{ese2_1}
}
\subfigure[]{
\includegraphics[scale=.2]{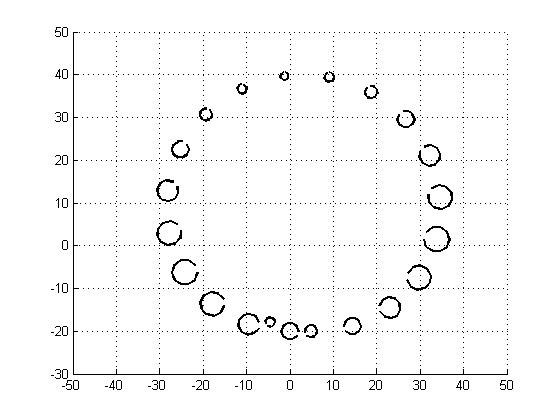}
\label{ese2_2}
}
\subfigure[]{
\includegraphics[scale=.2]{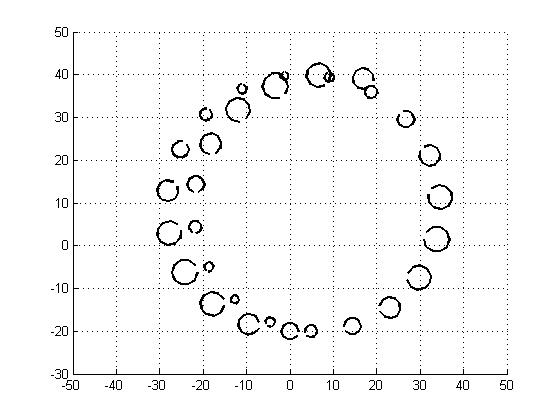}
\label{ese2_3}
}
\subfigure[]{
\includegraphics[scale=.2]{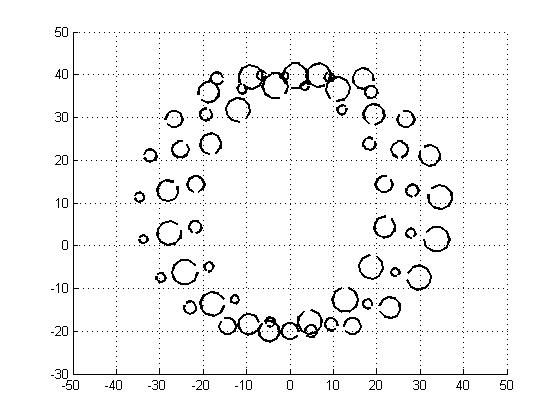}
\label{ese2_4}
}
\caption{
Various instances of a curve in $SE(2)$ with scaling.
}
\label{curve_ese2}
\end{figure}
We can construct a left-invariant riemannian metric
on this group.  It can be shown that for two elements 
$A_1, A_2$ in this group
\begin{equation}
d(A_1,A_2) = \sqrt{\alpha((\lambda_1-\lambda_2)^2+(\theta_1-\theta_2)^2)
+\beta \parallel d_1 - d_2 \parallel}
\label{ese2_metric}
\end{equation}
is a valid distance metric.
In Fig.\ref{curve_ese2}, a circular object under the action
of this group is shown for various time steps.
\begin{figure}[ht]
\centering
\subfigure[]{
\includegraphics[scale=.2]{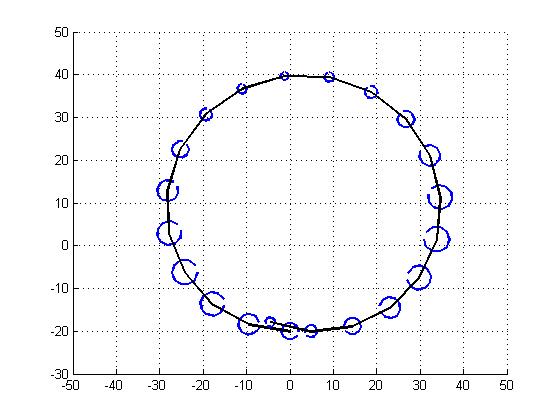}
\label{ese2_1}
}
\subfigure[]{
\includegraphics[scale=.2]{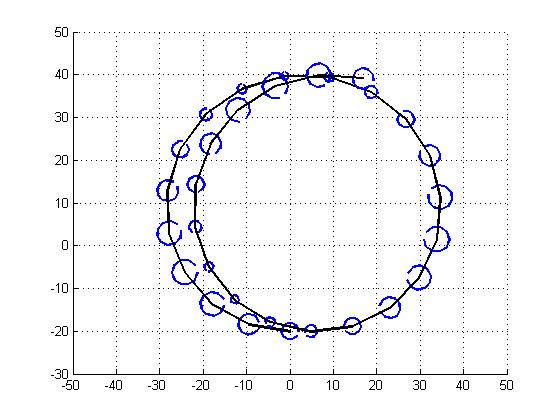}
\label{ese2_2}
}
\subfigure[]{
\includegraphics[scale=.2]{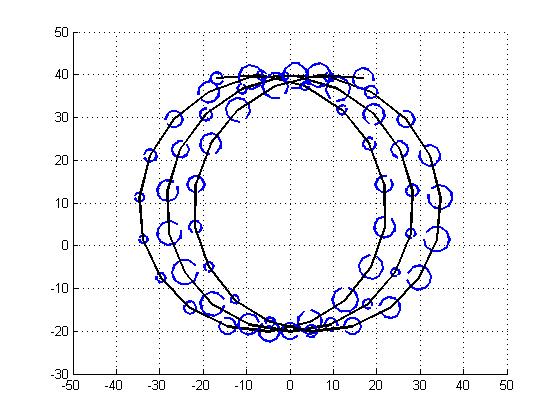}
\label{ese2_3}
}
\caption{
Instances of the reconstructed curve in $SE(3)$ with scaling.
}
\label{recon_curve_ese2}
\end{figure}
Assuming the curve is sampled densely, along with the distance 
measured by Eq.\ref{ese2_metric} we reconstruct the curve using 
MST.  The successfully reconstructed curve, with the values $\alpha=10$ 
and $\beta=1$, is shown in Fig.\ref{recon_curve_ese2}.
Important fact to note here is that the curve presented here is not 
a closed curve.  The algorithm is modified in this case to
take care of the end points.
In fact a simple nearest neighbour search will
also do the job of reconstruction once we give in the initial point.

\subsection{Application to video frame sequencing}
As an application of the curve reconstruction we take up a task of
ordering the frames $\{F_i\}_{i=1,\ldots,N}$ of a video sequence.  In Fig.\ref{unordered_video}
there are sixteen frames of a video sequence.  
\begin{figure}[h]
	\begin{center}
		\includegraphics[scale=.5]{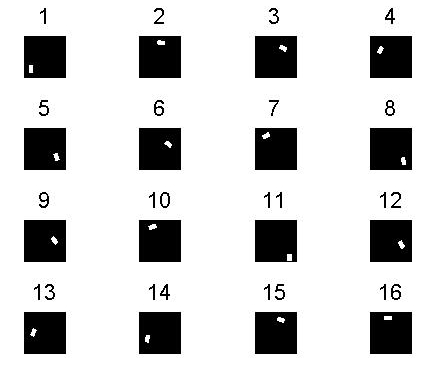}
	\end{center}
	\caption{Unordered video frames}
	\label{unordered_video}
\end{figure}
We use the rigid euclidean motion of an object in the frames as a clue
for re-ordering the frames.
Let us assume that the
object under observation is masked by a rectangle and it is segmented out of
the frames.  We also assume that the motion of the object is the
rigid body euclidean motion in $\mathbb{R}^2$.  
Further let the video 
frames from the sequence form a dense sample set of the motion curve.
As discussed in section \ref{metric_emg} we calculate the distances 
between frames as the distance between elements of $SE(2)$.
Although we do not focus on how to estimate the rotations we 
give a very primitive looking argument below to estimate the
distances between two frames.  And it turns out that the
estimates are good enough in this case to reconstruct the curve.
But in general we use the $[\theta,x,y]$ as the element
of $SE(2)$ and we assume that we have an oracle to give
these frame coordinates to the algorithm.

The euclidean distances between the means found out from the
relative positions of the rectangle is the first part of the
distance metric.
Next we estimate the rotation angle of the object with respect 
to a fixed inertial frame.
\begin{figure}[h]
	\begin{center}
		\includegraphics[scale=.5]{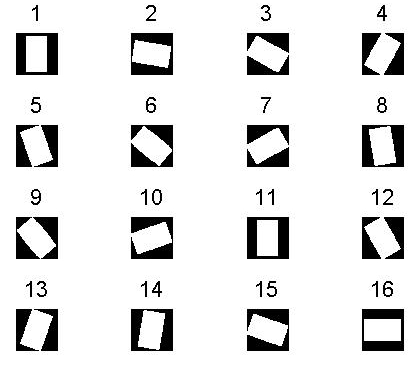}
	\end{center}
	\caption{Mean cancellation and rotation estimation}
	\label{rot_est}
\end{figure}
For this purpose first we register the objects with their means.
An observation reveals that if we overlap the registered rectangles
the area of the overlapping region provides a good estimate of
the rotation angle.  In fact as shown in Fig. for 
$\theta > \arctan(\frac{b}{a})$, the overlapped area is $\frac{a^2}{\sin{\theta}}$,
where $a$ is the shorter side of the rectangle.  Which clearly indicates
as $\theta$ increase the overlapping area decreases upto $\theta=\pi/2$.
For calculating the area we count the number of lattice points(pixels) inside
the overlapping regions.
\begin{figure}[h]
	\begin{center}
		\includegraphics[scale=.5]{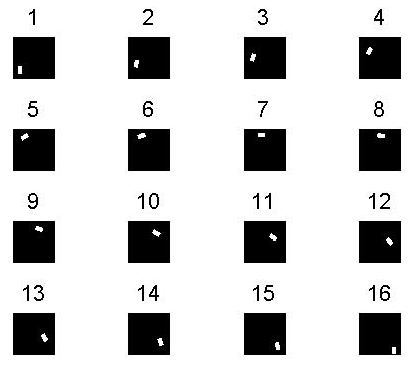}
	\end{center}
	\caption{Ordered video frames}
	\label{ordered_video}
\end{figure}
Finally with the estimate for $\theta$ combined with the euclidean distance
between means give the $d^*(F_1,F_2)$.
Using sequential search with known initial frame we re-order the frames see 
Fig.\ref{ordered_video}.
Even if we do not know about the initial frame, MST computes the correct
connections of the frames and gives a correct ordering upto end points.

Let us reconsider the distance metric on $SE(2)$ given
by Eqn.\ref{se2_metric}.  If we scale the three axis properly the
problem of curve reconstruction in $SE(2)$ reduces to the problem
of curve reconstruction in $\mathbb{R}^3$ and we may use
all the non-uniform sampling schemes and voronoi diagram based
reconstruction algorithms.  As an example we have used
NN-CRUST to reconstruct the curve above in the motion sequence
and we get the correct ordering as expected.

\section{Conclusion}
We showed that the MST gives the correct geodesic polygonal
approximation to the smooth, closed and simple curves in riemannian manifolds
if the sample is dense enough and we work inside the
injectivity radius.  We have worked out a conservative
bound on the uniform sampling of the curve.  The effect of
local topological behaviour of the underlying manifold 
was clearly identified and 
resolved by working inside the injectivity radius.
In general the scheme
works for the smooth arcs with endpoints also.  We have 
presented simulations for successfully
reconstructed curves in $SE(2)$ and $SE(3)$.  We have also
shown the applications of the combinatorial curve reconstruction 
for ordering motion frames in graphics and robotics.

If we work inside the  injectivity radius of the underlying manifold
we have taken care of the topological changes but to take care of geometric
changes we need to work inside the convexity radius as prescribed in
\cite{leibon00}.
We believe that the results of non uniform sampling for curves 
in $\mathbb{R}^n$
are transferable to the curves in riemannian manifold with careful modifications.
As an extension to this work
we would like to work out the necessary proofs and carry out simulations 
for supporting our belief.


%

\appendices
\section{Exponential and Logarithmic maps}
\label{ap1}
	\begin{res}
		Given $[\omega] \in so(3)$,
		\begin{equation}\exp[\omega]=I + \frac{\sin\Vert \omega \Vert}{\Vert \omega \Vert}\cdot[\omega]+
		\frac{1-\cos\Vert \omega \Vert}{\Vert \omega \Vert^2}\cdot [\omega]^2
		\end{equation}
	\end{res}
	
	\begin{res}
		Let $(\omega,v)\in se(3)$. Then 
		\begin{equation}
			\exp\left[\begin{array}{cc} 
			[\omega] & v\\
			0 & 0 \end{array}\right] = 	
			\left[\begin{array}{cc}
			\exp[\omega] & Av\\
			0 & 1 \end{array}\right]	
		\end{equation}
		where 
		\[
			A = I+\frac{1-\cos\Vert\omega\Vert}{\Vert\omega\Vert^2}\cdot [\omega]+
			\frac{\Vert\omega\Vert-\sin\Vert\omega\Vert}{\Vert\omega\Vert^3}\cdot[\omega]^2
		\]
	\end{res}

	\begin{res}
		Given $\theta \in SO(3)$ such that $Tr(\theta)\neq -1$. Then 
		\begin{equation}
			\log(\theta) = \frac{\phi}{2 \sin\phi}(\theta-\theta^T)		
		\end{equation}
		where $\phi$ satisfies $1+2 \cos\phi = Tr(\theta)$, $|\phi|<\pi$.  Further more,
		$\Vert \log\theta \Vert^2=\phi^2$.
	\end{res}
	
	\begin{res}
		Suppose $\theta \in SO(3)$ such that $Tr(\theta)\neq -1$, and let $b\in \mathbb{R}^3$.  
		Then
		\begin{equation}
			\log\left[\begin{array}{cc} 
			\theta & b\\
			0 & 1 \end{array}\right] = 	
			\left[\begin{array}{cc}
			[\omega] & A^{-1}b\\
			0 & 0 \end{array}\right]	
		\end{equation}		
		where $[\omega]=\log\theta$, and 
		\[
			A^{-1}= I-\frac{1}{2}\cdot[\omega]+\frac{2\sin\Vert\omega\Vert-\Vert\omega\Vert(1+\cos\Vert\omega\Vert)}
			{2\Vert\omega\Vert^2\sin\Vert\omega\Vert}\cdot[\omega]^2
		\]		
		
	\end{res}

	\begin{res}
		Let $\theta_1,\theta_2\in SO(3)$.  Then the distance $L=d(\theta_1,\theta_2)$
		induced by the standard bi-invariant metric on $SO(3)$ is 
		\begin{equation} d(\theta_1,\theta_2)=\Vert \log(\theta_1^{-1} \theta_2)\Vert\end{equation}
		where $\Vert\cdot \Vert$ denotes the standard Euclidean norm.
	\end{res}	
	\begin{res}
		Let $X_1=(\theta_1,b_1)$ and $X_2=(\theta_2,b_2)$ be two points in $SE(3)$.  Then the
		distance $L=d(X_1,X_2)$ induced by the scale dependent left-invariant metric on $SE(3)$
		is 
		\begin{equation}d(X_1,X_2)=\sqrt{c\Vert \log(\theta_1^{-1} \theta_2) \Vert^2+d\Vert b_2-b_1 \Vert^2}\end{equation}
		where $\Vert\cdot \Vert$ denotes the Euclidean norm.
	\end{res}


\ifCLASSOPTIONcompsoc
  \section*{Acknowledgments}
\else
  \section*{Acknowledgment}
\fi

The authors would like to acknowledge Prof. Gautam Dutta 
for discussions on the proof of the results in this
article.  The authors would also like to thank the
resource center DAIICT for providing references
needed for the work carried out.  
We acknowledge INRIA, Gamma researcher's team, 
\url{http://www-roc.inria.fr/gamma/gamma/disclaimer.php}, 
for their 3D-mesh files which we have used for 
simulations.
\ifCLASSOPTIONcaptionsoff
  \newpage
\fi



\bibliographystyle{IEEEtran}
\bibliography{IEEEabrv,ref}
\end{document}